\documentclass[conference,final]{IEEEtran}

\usepackage[cmex10]{amsmath}
\usepackage{booktabs}
\usepackage{amsfonts}
\usepackage{amssymb}
\usepackage{mathrsfs}
\usepackage{cite}
\usepackage{balance}

\usepackage[tight,footnotesize]{subfigure}
\usepackage{wasysym}
\usepackage{color}
\usepackage{epsfig} 
\usepackage{epstopdf}
\usepackage{float}
\usepackage{amsthm}
\usepackage{mathtools}
\allowdisplaybreaks
\newtheorem{lem}{Lemma}

\definecolor{light-gray}{gray}{0.8}

\def\nbb{{\mathbf{b}}}

\def\nbr{{\mathbf{r}}}

\def\nbx{{\mathbf{x}}}

\def\nb0{{\mathbf{0}}}
\def\nb1{{\mathbf{1}}}

\def\ncalA{{\mathcal{A}}}
\def\ncalB{{\mathcal{B}}}
\def\ncalC{{\mathcal{C}}}

\def\ncalG{{\mathcal{G}}}

\def\nbbP{{\mathbb{P}}}

\def\nrmd{{\rm d}}

\newtheorem{theorem}{Theorem}

\def\E{\mathbb{E}}

\def\pc{\mathtt{P_c}}

\def\R{\mathbb{R}}

\def\sir{\mathtt{SIR}}

\IEEEoverridecommandlockouts

\begin{document}
\title{Serving Distance and Coverage in a Closed Access PHP-Based Heterogeneous Cellular Network}
\author{\IEEEauthorblockN{Zeinab Yazdanshenasan, Harpreet S. Dhillon, and Peter Han Joo Chong
}
\thanks{H. S.~Dhillon  is with  Wireless@VT, Department of ECE, Virginia Tech, Blacksburg, VA, USA. Email: hdhillon@vt.edu. Z.~Yazdanshenasan and P. H. J.~Chong are with School of EEE, Nanyang Technological University (NTU), Singapore. Email: \{zeinab001, Ehjchong\}@e.ntu.edu.sg. The support of the US NSF (Grant CCF-1464293) is gratefully acknowledged.}
}
\maketitle
\thispagestyle{empty}
\pagestyle{empty}
\begin{abstract}
Heterogeneous cellular networks (HCNs) usually exhibit spatial separation amongst base stations (BSs) of different types (termed {\em tiers} in this paper). For instance, operators will usually not deploy a picocell in close proximity to a macrocell, thus inducing separation amongst the locations of pico and macrocells. This separation has recently been captured by modeling the small cell locations by a Poisson Hole Process (PHP) with the hole centers being the locations of the macrocells. Due to the presence of exclusion zones, the analysis of the resulting model is significantly more complex compared to the more popular Poisson Point Process (PPP) based models. In this paper, we derive a tight bound on the distribution of the distance of a typical user to the closest point of a PHP. Since the exact distribution of this distance is not known, it is often approximated in the literature. For this model, we then provide tight characterization of the downlink coverage probability for a typical user in a two-tier closed-access HCN under two cases: (i) typical user is served by the closest macrocell, and (ii) typical user is served by its closest small cell. The proposed approach can be extended to analyze other relevant cases of interest, e.g., coverage in a PHP-based open access HCN.%
\end{abstract}


\section{Introduction}
Stochastic geometry has emerged as a powerful tool for the modeling and analysis of HCNs~\cite{dhillon2012modeling,mukherjee2012sinr}.
The most popular approach is to model the BS locations of an HCN as a superposition of independent PPPs. While this lends tractability to the analysis of key performance metrics, the independent PPP-based model does not capture inherent spatial separation that exists amongst the locations of the BSs belonging to different tiers. In order to capture this inter-tier separation in a two-tier HCN, \cite{deng2014heterogeneous,dengheterogeneousax} have proposed a new model in which the macro BS locations are modeled by a PPP and the small cell BS locations are modeled around them using a PHP. The PHP model essentially places an exclusion zone of a given radius around each macro BS where small cells cannot lie. 
 
The presence of holes in the interference field makes it difficult to analyze PHP-based models. This has led to several approximations, including approximating a PHP by its baseline PPP, approximating it by a PPP with the same density as the PHP, and approximating it by a Poisson Cluster Process~\cite{ganti2006regularity,deng2014heterogeneous,dengheterogeneousax,Hybrid_ICC,lee2012interference}. In context of the HCN analysis, this means that the resulting approximate expressions for key metrics, such as coverage probability, are accurate only for a limited range of system parameters. Furthermore, the serving distance distribution in the presence of holes is not easy to characterize. As a result, prior works on the analysis of PHP-based HCNs either assumed the serving distances to be fixed or used curve-fitting to {\em fit} their distributions to a Weibull distribution~\cite{dengheterogeneousax}. More recently, we developed new tractable tools for the shot-noise analysis of an interference field modeled as a PHP by carefully {\em preserving the local-neighborhood} around the receiver location in order to maintain accuracy while approximating the far-field to lend tractability~\cite{yazdanshenasan2016poisson}. Using the same general methodology, we develop new tools for the analysis of PHP-based HCNs in this paper. Key contributions are summarized next.

{\em Contributions and outcomes.} We consider a two-tier closed-access PHP-based HCN model, where the locations of the macro BSs are modeled by a PPP and those of small cell BSs by a PHP. This naturally captures the inter-tier spatial separation. To enable the downlink analysis of this model, we first derive a bound on the distance of a typical user to the closest point of a PHP using general approach of {\em preserving the  local neighborhood} proposed by the authors in~\cite{yazdanshenasan2016poisson}. Numerical comparisons show that the bound is remarkably tight across wide range of system parameters. This gives us the serving distance distribution when a typical user connects to the small cell tier. We then derive several new bounds and approximations for the coverage probability for the closed access case where the typical user is authorized to connect to one of the two tiers (either macro or small cells).

          \begin{figure}[t!]
  \centering{
              \includegraphics[width=.9\linewidth]{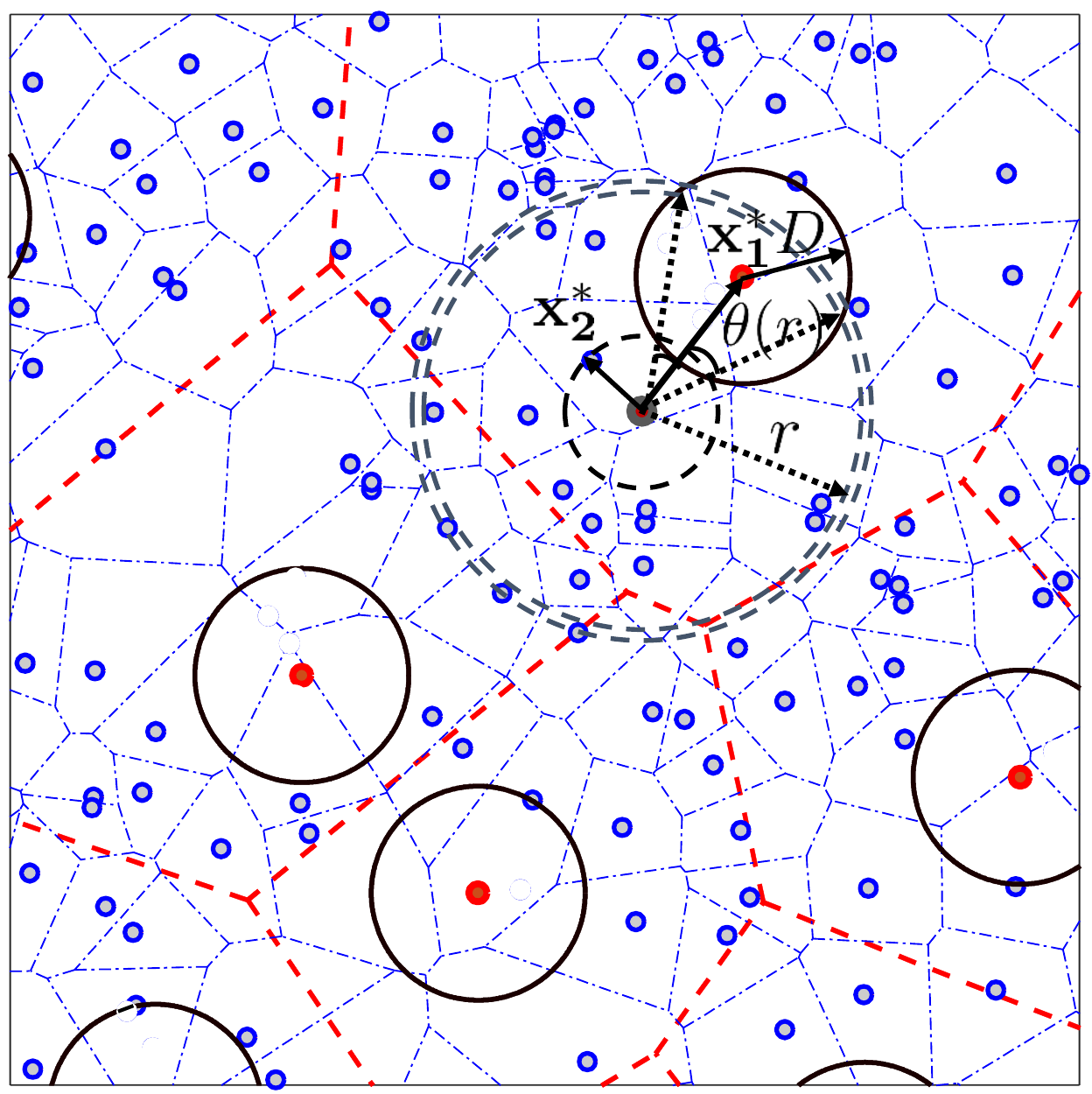}
              \caption{Illustration of the PHP-based two-tier HCN system model.}
\label{Serv_dis00}}
          \end{figure} 
\section{System model}
\label{Systmodel}
We consider a two-tier HCN consisting of macrocells and small cells. The macro BS locations are modeled as a homogenous PPP $\{{\bf x}_1\} \equiv \Phi_1$ with density $\lambda_1$ BSs per meter$^2$. As proposed in~\cite{deng2014heterogeneous,dengheterogeneousax}, we capture the spatial separation among the small cell and macro BSs by modeling the small cell locations as a PHP, which can be defined formally as: 
\begin{align}\label{PHPdef1}
\Phi_2 =\{\nbx_2\in \Omega : \nbx_2\notin \Xi_D\}= \Omega \setminus \Xi_D,
\end{align}
where  $\Omega$ is a baseline PPP of density $\lambda_2$ and $ \Xi_D \triangleq  \bigcup_{\nbx_1 \in \Phi_1} \mathbf{b}(\nbx_1,D)$ with $\mathbf{b}(\nbx_1,D)$ being a ball of radius $D$ centered at $\nbx_1$. In other words, $\Phi_2$ is generated by carving out holes of fixed radius $D$ centered around the points of $\Phi_1$ from $\Omega$. This means the minimum separation between a macro and small cell is $D$. This setup is illustrated in Fig.~\ref{Serv_dis00}.
 Please refer to~\cite{yazdanshenasan2016poisson} for a more detailed discussion about PHP. 

For simplicity of exposition, we consider macro and small cell users separately. We assume that each macro user is served by its closest macrocell and each small cell user is served by its closest small cell. While this case may appear to be restrictive, its analysis entails significant complexity due to the presence of {\em holes} in $\Phi_2$. The tools developed in this paper to handle this complexity also provide first few concrete steps towards the analysis of open access case, which is significantly more complex due to the consideration of cell selection jointly across the two tiers. The results for the open access case will be provided in the extended journal version of this paper. 

The user locations are modeled by an independent stationary point process. The downlink analysis is performed at a typical user located at the origin under two cases: (i) it is served by its closest small cell, and (ii) it is served by its closest macrocell. Denoting the distance from the typical user to  its closest BS from the $k^{th}$ tier by $\|{\bf x}_{k}^*\| = Z_{k}$ and transmit power of the $k^{th}$ tier BSs by $P_k$,  the received power at the typical user is:  $P_{r,k} = P_k h_{\nbx_k^*} {Z_{k}}^{-\alpha}$, where 
$h_{\nbx_k^*}\sim \exp(1)$ models Rayleigh fading and $\alpha>2$ is path-loss exponent. The  signal to interference ratio ($\sir$) at the typical user served by a BS of the $k^{th}$ tier is: 
\begin{align*}
\sir(\|{\bf x}_{k}^*\|)\!=\!\frac{P_k h_{\nbx_k^*} {\|{\bf x}_{k}^*\| }^{-\alpha} }{\sum_{j=1}^2\sum_{\substack{\nbx_j \in {\Phi_{j}\setminus\{\nbx_{k}^*\}}}}P_{j} h_{\nbx} \|\nbx_j\|^{-\alpha}}\!,
\end{align*}
where $k\in\{1,2\}$.
The thermal noise is assumed to be negligible compared to the interference and is hence ignored. 
\section{Coverage Probability}
This is the main technical section of the paper where we first derive the distribution of the distance from the typical user to its closest macro and small cell BSs (denoted by $Z_1$ and $Z_2$). Using these distributions, we will derive expressions for the coverage probability of the typical user in the two cases. 
\label{SIRCoverPrb}
\subsubsection{Distributions of serving distances $Z_1$ and $Z_2$} 
Recall that $Z_1$ denotes the distance of the typical user (placed at the origin) to its closest macrocell (closest point of $\Phi_1$). The distribution of $Z_1$ is well-known and can be derived using the null probability of a PPP~\cite{haenggi2005distances}. For this case, it comes out to be%
\begin{align}	\label{Eq: clds00}
f_{Z_1}(z_1)=2\pi\lambda_1 z_1\exp(-\pi\lambda_1 {z_1}^2), \quad z_1 > 0.
\end{align}
Also recall that $Z_2$ denotes the distance of the typical user to its closest small cell. Its distribution is much more complex to characterize. In order to curtail the complexity, we consider only the closest hole to the typical user in this derivation and ignore all other holes. In other words, for the derivation of this distance distribution, PHP $\Phi_2$ is approximated by $\Omega$ with only one (closest) hole carved out. This approach clearly underestimates $Z_2$ and the resulting random variable is denoted by $\hat{Z_2}$.  
The distribution of distance $\hat{Z_2}$ is given by Lemma~\ref{Lemma1_Ass2}. 
The proof is provided in Appendix~\ref{app: A}.
\begin{figure*}[!t]
\begin{lem}
\label{Lemma1_Ass2} The probability density function (PDF) of distance $\hat{Z_2}$ is given by
\begin{align}	\label{Eq: FRcrofton}
f_{\hat{Z_2}|Z_1\leq D}(\hat{z_2})= \left\{\begin{matrix}
(2\pi\lambda_2 {\hat{z_2}}-\lambda_2\frac{\nrmd A_\mathrm{{ins}}({\hat{z_2}},Z_1)}{\nrmd \hat{z_2}})\exp(-\lambda_2 (\pi{\hat{z_2}}^2-A_\mathrm{{ins}}({\hat{z_2}},Z_1))) & D-Z_1 < {\hat{z_2}}\leq Z_1+D\\ 
2\pi\lambda_2 {\hat{z_2}}\exp(-\pi\lambda_2 ({\hat{z_2}}^2-D^2)) & {\hat{z_2}}> Z_1+D
\end{matrix}\right.
\end{align} and
\begin{align}	\label{Eq: FRcrofton2}
f_{\hat{Z_2}|Z_1> D}(\hat{z_2})= \left\{\begin{matrix}
2\pi\lambda_2 {\hat{z_2}}\exp(-\pi\lambda_2 {\hat{z_2}}^2) & {\hat{z_2}}\leq Z_1-D\\ 
(2\pi\lambda_2 {\hat{z_2}}-\lambda_2\frac{\nrmd A_\mathrm{{ins}}({\hat{z_2}},Z_1)}{\nrmd \hat{z_2}})\exp(-\lambda_2 (\pi{\hat{z_2}}^2-A_\mathrm{{ins}}({\hat{z_2}},Z_1))) & Z_1-D < {\hat{z_2}}\leq Z_1+D\\ 
2\pi\lambda_2 {\hat{z_2}}\exp(-\pi\lambda_2 ({\hat{z_2}}^2-D^2)) & {\hat{z_2}}> Z_1+D
\end{matrix}\right.
\end{align}
where $\hat{z_2}$ and $Z_1$ denote the distance from the typical user to its serving small cell BS and its distance to the closest macro BS,
 respectively. Further $A_\mathrm{{ins}}({\hat{z_2}},Z_1)=\ncalB(D,\hat{z_2},Z_1)+\ncalB(\hat{z_2},D,Z_1)-0.5\ncalA(D,\hat{z_2},Z_1)$ denotes the area of intersection between two circles $\nbb(0,\hat{z_2})$ and $\nbb(Z_1,D)$.
Other functions appearing in the above expression are as follows: $\ncalA(\kappa,\zeta,\eta)=\sqrt{\left(\kappa^2-(\eta-{\zeta})^2\right) \left(({\zeta}+\eta)^2-\kappa^2\right)}$, $\ncalB(\kappa,\zeta,\eta)=\kappa^2 \arccos \left(\frac{\kappa^2-{\zeta}^2+\eta^2}{2 \kappa \eta}\right)$ and $\ncalC(\kappa,\zeta,\eta)=\frac{\kappa^2-\zeta^2+\eta^2}{2 \kappa^2 \eta}$.
\end{lem}
  \noindent\rule{18.5cm}{0.4pt}
\end{figure*}

Using these results, we now focus on the derivation of coverage probability in the two cases: (i) typical user served by the macro tier ($\pc_1$), and (ii) typical user served by the small cell tier ($\pc_2$). Coverage probability $\pc_k$ is defined as the probability that the instantaneous $\sir$ of a typical user when it is served by the $k^{th}$ tier BS (denoted by $\sir({Z_{k}})$) is greater than a target $\sir$ $\gamma$. It is mathematically defined as 
\begin{align}\label{CovgEchTi}
\pc_k&=\E_{Z_{k}}\left[ \,\mathbb{P}\{\sir({z_{k}})\geq \gamma\,|{Z_{k}}\}\right]\\ \nonumber&{=}
\int_{{z_{k}}>0}\,\mathbb{P}\{\sir({z_{k}})\geq \gamma\}f_{Z_{k}}({z_{k}})\nrmd {z_{k}}.
\end{align}
Key challenge in the derivation of $\pc_k$ is the presence of holes (exclusion zones) in the PHP $\Phi_2$. Recently, we developed new tools to handle these exclusion zones in an ad hoc network modeled by a PHP in~\cite{yazdanshenasan2016poisson}. We extend these tools to derive remarkably right bounds and approximations for the coverage probability $\pc_k$. More details are provided next.

\subsubsection{Coverage probability when the typical user is served by its closest macro cell}
In this case, the 
set of BSs contributing to interference are macro BSs from $\Phi_1\setminus \nbx_1^*$ and small BSs from $\Phi_2$.
We first derive the coverage probability of the typical user by considering only the closest hole to the typical user.
Thereby, the interference field of small cell BSs is overestimated by $\Omega \cap \nbb^c(\nbx_{1}^*,D)$. 
This provides a lower bound on the coverage probability, which is derived next. 
\begin{theorem}\label{ThmA1} 
Considering only the closest hole in the interference field of small cells, the coverage probability of a typical user when it connects to its closest macro BS is
\begin{align}
\pc_1&\geq \int_{0}^{D}\ncalG_{1}^{(1)}(s)\ncalG_{2}^{(1)}(s,z_1)\ncalG_{3}^{(1)}(s,z_1)f_{Z_1}(z_1)\nrmd z_1\\ \nonumber&{+}\int_{D}^{\infty}\hat{\ncalG}_{1}^{(1)}(s)\ncalG_{2}^{(1)}(s,z_1)\ncalG_{3}^{(1)}(s,z_1)f_{Z_1}(z_1)\nrmd z_1,
\end{align}
where $s=\frac{\gamma {z_1}^{\alpha}}{P_1}$,
$\ncalG_{1}^{(1)}(s)=\exp\left(\int_{D-z_1}^{\infty}\frac{-2\pi\lambda_2 r\mathrm{d}r}{1+{\frac{r^{\alpha}}{sP_2}}}\right)$, $\hat{\ncalG}_{1}^{(1)}(s)=
\exp\left(-\pi\lambda_2 \frac{(sP_2)^{2/\alpha}}{\mathrm{sinc}(2/\alpha)}\right)$, $\ncalG_{2}^{(1)}(s,z_1)= \exp\left(\int_{|{z_1-D}|}^{{z_1+D}}{\arccos\left(\frac{r^2+z_1^2-D^2}{2z_1r}\right)}\frac{2\lambda_2}{1+\frac{r^\alpha}{P_2s}} r\mathrm{d}r\right)$,  $\ncalG_{3}^{(1)}(s,z_1)=
\exp\Big(-2\pi\lambda_1 \int_{{z_1}}^{\infty} \Big({1-\frac{1}{1+s P_1 v^{-\alpha}}}\Big)v\nrmd v\Big)$.
\end{theorem}
\begin{proof}
See Appendix \ref{app: B}.
\end{proof}
The tightness of this bound will be demonstrated in the numerical results section (see Fig.~\ref{Covg_MF00}). 

We now extend the above approach to incorporate all the holes in the interference field. In order to maintain tractability, we ignore the effect of overlaps among the holes such that each hole is assumed to be carved out individually from the baseline PPP instead of carving out the union of the holes.
Some of the points of the baseline PPP located inside the holes may be virtually removed multiple times by using this approach, thus leading to underestimation of the interference power experienced by the typical user from small cell BSs. 
This leads to an upper bound on the coverage probability. Please refer to~\cite{yazdanshenasan2016poisson} for more details where we proposed this bounding technique for the analysis of ad hoc networks.  

\begin{theorem}\label{ThmA2} 
The coverage probability of a typical user which is served by its closest macro BS is upper bounded by 
\begin{align}
\pc_{1}&\leq \int_{0}^{D}\ncalG_{1}^{(1)}(s)\ncalG_{2}^{(1)}(s,z_1)\ncalG_{4}^{(1)}(s,z_1)f_{Z_1}(z_1)\nrmd z_1
\nonumber\\&+
\int_{D}^{\infty}\hat{\ncalG}_{1}^{(1)}(s)\ncalG_{2}^{(1)}(s,z_1)\ncalG_{4}^{(1)}(s,z_1)f_{Z_1}(z_1)\nrmd z_1
\end{align}
where $s=\frac{\gamma {z_1}^{\alpha}}{P_1}$, 
$\ncalG_{1}^{(1)}(s)=\exp\!\left(\!\int_{D-z_1}^{\infty}\!\frac{-2\pi\lambda_2 r\mathrm{d}r}{1+{\frac{r^{\alpha}}{sP_2}}}\!\right)$, $\hat{\ncalG}_{1}^{(1)}(s)=
\exp\left(-\pi\lambda_2 \frac{(sP_2)^{2/\alpha}}{\mathrm{sinc}(2/\alpha)}\right)$, $\ncalG_{2}^{(1)}(s,z_1)= \exp\left(f(s,z_1)\right)$,  $f(s,z_1)=\int_{|{z_1-D}|}^{{z_1+D}}{\arccos\left(\frac{r^2+z_1^2-D^2}{2z_1r}\right)}\frac{2\lambda_2}{1+\frac{r^\alpha}{P_2s}} r\mathrm{d}r$, 
\small
$\ncalG_{4}^{(1)}(s,z_1)=
\exp\Big(-2\pi\lambda_1 \int_{{z_1}}^{\infty} \Big({1-\exp\big(f(s,v)\big)\zeta(s,v)}\Big)v\nrmd v\Big)$,
\normalsize
$\zeta(s,v)=\frac{1}{1+s P_1 v^{-\alpha}}$.
\end{theorem}
\begin{proof}
See Appendix \ref{app: C}.
\end{proof}
\subsubsection{Coverage probability when the typical user is served by its closest small cell}
In this case, the set of interfering BSs constitutes macro BSs from $\Phi_1$ and small cell BSs from $\Phi_2\setminus \nbx_2^*$ where $\nbx_2^*$ denotes the serving small cell BS. As was the case in Theorem~\ref{ThmA1}, we again consider only the effect of the closest hole to the typical user, thus overestimating the interference field. For the distribution of serving distance $Z_2$, we use its approximation derived in Lemma~\ref{Lemma1_Ass2}. The resulting approximation for the coverage probability is provided next.
\begin{theorem}\label{ThmA3} 
Considering only the closest hole in the interference field, the coverage probability of the typical user which connects to its closest small cell BS is approximated by
\begin{align}
&\pc_2{\simeq}
\!\!\int_{0}^{D}\int_{D-z_1}^{\infty}\!\!\!\!\ncalG_{1}^{(2)}(s,\hat{z_2})\ncalG_{2}^{(2)}(s,z_1,\hat{z_2})\ncalG_{3}^{(2)}(s,z_1)\ncalG_{4}^{(2)}(s,z_1)\nonumber\\&f_{\hat{Z_2}|Z_1\leq D}(\hat{z_2})
f_{Z_1}(z_1)\nrmd \hat{z_2}\nrmd z_1
+\nonumber\\&
\!\!\int_{D}^{\infty}\int_{0}^{\infty}\!\!\!\!\ncalG_{1}^{(2)}(s,\hat{z_2})\ncalG_{2}^{(2)}(s,z_1,\hat{z_2})\ncalG_{3}^{(2)}(s,z_1)\ncalG_{4}^{(2)}(s,z_1)
\nonumber\\&
f_{\hat{Z_2}|Z_1>D}(\hat{z_2})f_{Z_1}(z_1)\nrmd \hat{z_2}\nrmd z_1
\end{align}
where $s=\frac{\gamma {\hat{z_2}}^{\alpha}}{P_2}$ , $\ncalG_{1}^{(2)}(s,\hat{z_2})=\exp\left(\int_{\hat{z_2}}^{\infty}\frac{-2\pi\lambda_2 r\mathrm{d}\nbr}{1+{\frac{r^{\alpha}}{sP_2}}}\right)$, $\ncalG_{2}^{(2)}(s,z_1,\hat{z_2})= \exp\left(g(s,z_1,\hat{z_2})\right)$,  $g(s,z_1,\hat{z_2})=\exp\!\left({\int_{\max(\hat{z_2},|z_1-D|)}^{\max(\hat{z_2},z_1+D)}\!\arccos\left(\frac{r^2+
z_1^2-D^2}{2z_1r}\right)\frac{2\lambda_2}{1+\frac{r^{\alpha}}{sP_2}} r \nrmd r}\!\right)$,  $\ncalG_{3}^{(2)}(s,z_1)=\zeta(s,z_1)=\frac{1}{1+s P_1 z_1^{-\alpha}}$, $\ncalG_{4}^{(2)}(s,z_1)=\\
\exp\left(\!\!-2\pi\lambda_1 \!\int_{{z_1}}^{\infty} \!\!\left({1- \zeta(s,v)}\right)v\nrmd v\right)$.
\end{theorem}
\begin{proof}
See Appendix \ref{app: D}.
\end{proof}
The tightness of this approximation will be demonstrated in the numerical results section (see Fig.~\ref{Covg_MF00}). On the same lines as Theorem~\ref{ThmA2}, we now provide another approximation for the coverage probability where we underestimate the interference power from the small cells by carving out holes in a PHP individually without considering overlaps amongst them. The resulting approximation is provided next.
\begin{theorem}\label{ThmA4} 
The coverage probability of a typical user which is served by its closest small cell BS is approximated by
\begin{align}
&\pc_2{\simeq}
\!\!\int_{0}^{D}\int_{D-z_1}^{\infty}\!\!\!\!\ncalG_{1}^{(2)}(s,\hat{z_2})\ncalG_{2}^{(2)}(s,z_1,\hat{z_2})\ncalG_{3}^{(2)}(s,z_1)\nonumber\\&\ncalG_{5}^{(2)}(s,z_1,\hat{z_2})f_{\hat{Z_2}|Z_1\leq D}(\hat{z_2})
f_{Z_1}(z_1)\nrmd \hat{z_2}\nrmd z_1
+\nonumber\\&
\!\!\int_{D}^{\infty}\int_{0}^{\infty}\!\!\!\!\ncalG_{1}^{(2)}(s,\hat{z_2})\ncalG_{2}^{(2)}(s,z_1,\hat{z_2})\ncalG_{3}^{(2)}(s,z_1)\ncalG_{5}^{(2)}(s,z_1,\hat{z_2})
\nonumber\\&
f_{\hat{Z_2}|Z_1>D}(\hat{z_2})f_{Z_1}(z_1)\nrmd \hat{z_2}\nrmd z_1
\end{align}
where $s=\frac{\gamma {\hat{z_2}}^{\alpha}}{P_2}$ , $\ncalG_{1}^{(2)}(s,\hat{z_2})=\exp\left(\int_{\hat{z_2}}^{\infty}\frac{-2\pi\lambda_2 r\mathrm{d}\nbr}{1+{\frac{r^{\alpha}}{sP_2}}}\right)$, $\ncalG_{2}^{(2)}(s,z_1,\hat{z_2})= \exp\left(g(s,z_1,\hat{z_2})\right)$,  $g(s,z_1,\hat{z_2})=\exp\!\left({\int_{\max(\hat{z_2},|z_1-D|)}^{\max(\hat{z_2},z_1+D)}\!\arccos\left(\frac{r^2+
z_1^2-D^2}{2z_1r}\right)\frac{2\lambda_2}{1+\frac{r^{\alpha}}{sP_2}} r \nrmd r}\!\right)$,  $\ncalG_{3}^{(2)}(s,z_1)=\zeta(s,z_1)=\frac{1}{1+s P_1 z_1^{-\alpha}}$, $\ncalG_{5}^{(2)}(s,z_1)=\\
\exp\left(\!\!-2\pi\lambda_1 \!\int_{{z_1}}^{\infty} \!\!\left({1- \exp\big(g(s,v)\big)\zeta(s,v)}\right)v\nrmd v\right)$.
\end{theorem}
\begin{proof} 
See Appendix \ref{app: E}
\end{proof}
          \begin{figure}[t!]
  \centering{
              \includegraphics[width=.9\linewidth]{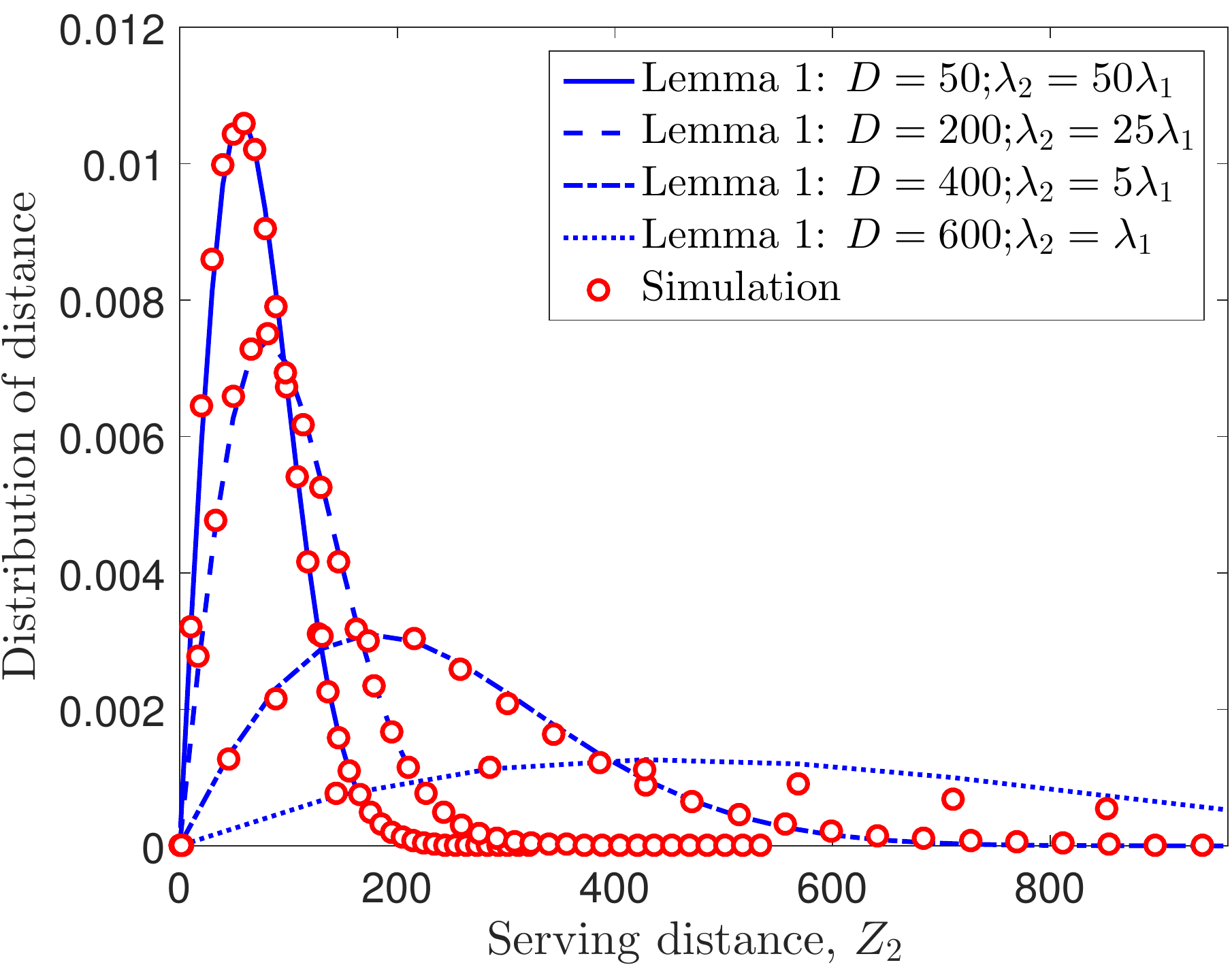}
              \caption{Distribution of the distance from the typical user to its closest small cell BS in a PHP. }
\label{DisServ_00}}
          \end{figure}
          \begin{figure}[t!]
  \centering{
              \includegraphics[width=.9\linewidth]{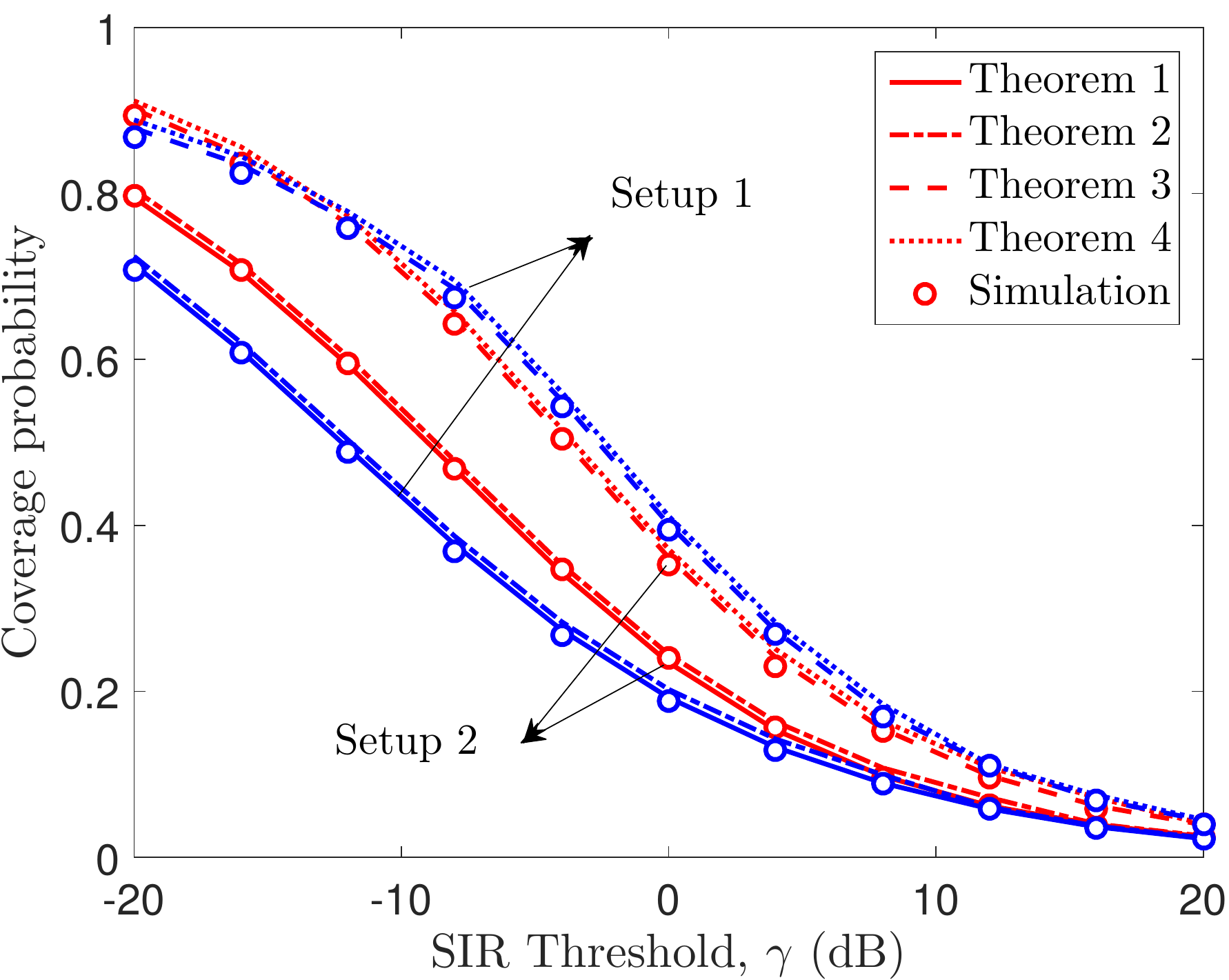}
              \caption{Comparison of macro and small cell coverage probability results.}
\label{Covg_MF00}}
          \end{figure}
\subsubsection{Discussion and Numerical Results}
\label{NumerReslt}
In this section, we validate our results by comparing them against results obtained from Monte Carlo simulations. 
We consider two setups such that in both cases, $\lambda_1=1$ BS/${\mathrm{km}}^2$, $\alpha=4$. 
In setup 1, $\lambda_2=50\lambda_1$, $P_1=10^3P_2$ and $D=50$ m, 
while in the setup 2, $\lambda_2=25\lambda_1$, $P_1=10^2P_2$ and $D=200$ m.
As illustrated in Fig.~\ref{DisServ_00}, the distribution of the distance $\hat{Z_2}$ derived in Lemma \ref{Lemma1_Ass2} which underestimated the true serving distance $Z_2$ is surprisingly tight for all configurations (even when the density of small cells is low and the radius of exclusion zone $D$ is large). 
Fig.~\ref{Covg_MF00} presents the coverage probability of typical user under two cases: (i) served by macrocell, and (ii) served by small cell, as a function of the $\sir$ threshold. Here, we plot the four analytical results given by Theorems \ref{ThmA1}- \ref{ThmA4} and compare them with results obtained from simulations. We note that all the analytical results are surprisingly tight. As expected, the coverage probability of the users served by the small cells in this setup is higher (due to lower serving distance). 
\section{Conclusion}  
In this paper, we developed a new approach to the downlink coverage probability analysis of a two-tier HCN modeled as a PHP. The proposed approach is based on the tools developed by the authors for the analysis of PHPs in \cite{yazdanshenasan2016poisson,YazDhiC2016}. The key idea is to curtail the complexity by preserving the local neighborhood of the interference field around the typical point while simplifying the far-field of interference. Using this approach, we derived a new and remarkably tight bound for the serving distance of a typical user in a PHP to its closest small cell. We then derived coverage probability for a typical user under two closed-access cases: (i) typical user is served by macro tier, and (ii) typical user is served by a small cell tier. A meaningful extension of this work is to derive the downlink coverage probability for open-access PHP-based HCNs. 
\appendix
\subsection{Proof of Lemma~\ref{Lemma1_Ass2}}
\label{app: A}
We consider only the closest hole to the typical user, thus leading to a lower bound on the serving distance $Z_2$ which is denoted by $\hat{Z_2}$.
Conditioned on the distance 
$Z_1$, 
the distribution of $\hat{Z_2}$  is
\begin{align*}
&\nbbP(\hat{Z_2} > \hat{z_2}|Z_1)\!=\!\nbbP(\text {Number of points of $\Phi_2$ in the set}\ \\\nonumber &\{\nbb(0,\hat{z_2})\} = 0)= \exp\left(-\lambda_2 \left(\pi \hat{z_2}^2-A_\mathrm{{ins}}({\hat{z_2}},Z_1)\right)\right),~ z_1 > 0.
\end{align*}
where $\hat{z_2}$ and $Z_1$ denote the distance from the typical user to its serving small cell and its distance to the closest macro BS. 
Further, $A_\mathrm{{ins}}({\hat{z_2}},Z_1)$ denotes  the area of intersection between two circles $\nbb(0,\hat{z_2})$ and $\nbb(Z_1,D)$.
The result now follows by differentiating the above expression.
\subsection{Proof of Theorem~\ref{ThmA1}}
\label{app: B}
Coverage probability conditioned on the serving distance when $Z_1\geq D$ is 
\begin{align}
&\pc_{{1}| Z_1\geq D }\!\geq
\E \Bigg[\!\exp\! \Bigg(-s\Bigg(\sum_{\nbx_2 \in \Omega \cap \nbb^c(\nbx_1^*,D)}\!\!\!P_2 h_{\nbx_2} \|\nbx_2\|^{-\alpha}\nonumber \\&+\sum_{\nbx_1\in \Phi_1/{\nbx_1^*}}P_1 h_{\nbx_1} \|\nbx_1\|^{-\alpha}\Bigg)\!\! \Bigg)\! \Bigg]\stackrel{\text{(a)}}{=}
\nonumber \\&
\E_{\Phi_1|Z_1} \Bigg[\!\exp\left(-\lambda_2 \left(\int_{\R^2 }\frac{\mathrm{d}\nbx_2}{1+{\frac{\|\nbx_2\|^{\alpha}}{sP_2}}}- \int_{\mathbf{b}(\nbx_1^*,D)}\frac{\mathrm{d}\nbx_2}{1+{\frac{\|\nbx_2\|^{\alpha}}{sP_2}}}\right)\!\right)\nonumber \\& \times  \prod_{\nbx_1 \in \Phi_\mathrm{1}/{\nbx_1^*}}\!\!\zeta(s,\|\nbx_1\|)\Bigg]
\stackrel{\text{(b)}}{=}\hat{\ncalG}_{1}^{(1)}(s)\times \exp\left(f(s,z_1)\right) \times
\nonumber\\&
\exp\Bigg(-2\pi\lambda_1 \int_{{z_1}}^{\infty} \Big({1-\zeta(s,v)}\Big)v\nrmd v\Bigg)
\end{align}
where $\zeta(s,v)=\frac{1}{1+s P_1 v^{-\alpha}}$. Step (a) follows from taking expectations over channel gains $h_{\nbx_k} \sim \exp(1)$ and by using the PGFL of the PPP\cite{chiu2013stochastic},
 $\hat{\ncalG}^{(1)}(s)=\exp\left(-\pi\lambda_2 \frac{(sP)^{2/\alpha}}{\mathrm{sinc}(2/\alpha)}\right)$, as shown in~\cite[Appendix B]{dhillon2012modeling} and the second term in (b) follows from conditioning on the $Z_1$ using the cosine-law (as shown in Fig.~\ref{Serv_dis00}): $r^2+\|\nbx_1^*\|^2- 2 r\|\nbx_1^*\| \cos\theta(r)=D^2$ along with converting form Cartesian to polar coordinates by substituting $f(s,z_1)=\int_{|{z_1-D}|}^{{z_1+D}}{\arccos\left(\frac{r^2+z_1^2-D^2}{2z_1r}\right)}\frac{2\lambda_2}{1+\frac{r^\alpha}{P_2s}} r\mathrm{d}r$ which is the effective interference power ``removed'' by the closest hole. 
It should be noted that the proof of $\pc_{{1}|Z_1<D}$ is similar to the proof of $\pc_{{1}|Z_1\geq D}$ except that the closest point of $\Phi_2$ is at least a distance
$D-z_1$ away from the typical point located inside the hole, and hence leads to  $\ncalG_{1}^{(1)}(s)=\exp\!\left(\!\int_{D-z_1}^{\infty}\!\frac{-2\pi\lambda_2 r\mathrm{d}r}{1+{\frac{r^{\alpha}}{sP_2}}}\!\right)$.
The final result follows from deconditioning over serving distance $Z_1$ whose distribution is given by (\ref{Eq: clds00}). 
\subsection{Proof of Theorem~\ref{ThmA2}}
\label{app: C}
The coverage probability in this case conditioned on the serving distance when $Z_1\geq D$ is 
\begin{align}
&
\pc_{{1}| Z_1 \geq D }\!=
\E \Bigg[\!\exp\! \Bigg(\!\!-s\Bigg(\sum_{\nbx_2\in \Phi_2}\!\!\!P_2 h_{\nbx_2} \|\nbx_2\|^{-\alpha}
\nonumber \\&+\!\!\!\!\sum_{\nbx_1\in \Phi_1/{\nbx_1^*}}\!\!\!\!\!\!\! P_1 h_{\nbx_1} \|\nbx_1\|^{-\alpha}\Bigg)\!\! \Bigg)\! \Bigg]
\nonumber \\&
\stackrel{\text{(a)}}{=} \E_{\Phi_1|Z_1} \Bigg[\!\exp\left(\!-\lambda_2 \left(\int_{\R^2 }\frac{\mathrm{d}\nbx_2}{1+{\frac{\|\nbx_2\|^{\alpha}}{sP_2}}}- \int_{\Xi_D}\frac{\mathrm{d}\nbx_2}{1+{\frac{\|\nbx_2\|^{\alpha}}{sP_2}}}\!\right)\!\right) \nonumber \\&  \times \prod_{\nbx_1 \in \Phi_\mathrm{1}/{\nbx_1^*}}\!\!\zeta(s,\|\nbx_1\|)\Bigg]
\stackrel{\text{(b)}}
{\leq} \hat{\ncalG}_{1}^{(1)}(s)\times \E_{\Phi_1|Z_1}\Bigg[\nonumber \\&
\exp\Bigg(\lambda_2 \Bigg(\sum_{\nbx_1 \in \Phi_1}\!\! \int_{\nbb(\nbx_1, D)} 
\frac{\mathrm{d}\nbx_2}{1+{\frac{\|\nbx_2\|^{\alpha}}{sP_2}}} \Bigg)\Bigg)\!\!\!\prod_{\nbx_1 \in \Phi_\mathrm{1}/{\nbx_1^*}}\!\!\zeta(s,\|\nbx_1\|)\Bigg]
\nonumber \\&  
\stackrel{\text{(c)}}{=}\ncalG_{1}^{(1)}(s)\exp\left(f(s,\|\nbx_1^*\|)\right)\E_{\Phi_1|Z_1} \Bigg[\prod_{\nbx_1 \in \Phi_\mathrm{1}/{\nbx_1^*}}\!\!\!\exp\left(f(s,\|\nbx_1\|)\right)\nonumber \\& \zeta(s,\|\nbx_1\|)\!\Bigg]
\stackrel{\text{(d)}}{=}\ncalG_{1}^{(1)}(s)\times \exp\left(f(s,z_1)\right) \times
\nonumber\\&
\exp\!\Bigg(\!\!-2\pi\lambda_1 \int_{{z_1}}^{\infty} \Big({1-\exp\big(f(s,v)\big)\times\zeta(s,v)}\Big)v\nrmd v\!\Bigg)
\end{align}
where  $\Phi_2=\Omega\!\setminus\!\Xi_D$,  $\Xi_D$ in step (a) is $\triangleq  \bigcup_{\nbx_1 \in \Phi_1} \mathbf{b}(\nbx_1,D)$ as defined in (\ref{PHPdef1}), $\zeta(s,u)=\frac{1}{1+s P_1 u^{-\alpha}}$, (a) comes from taking expectations over channel gains $h_{\nbx_k} \sim \exp(1)$ and by using the PGFL of the PPP $\Omega$ given $\Xi_D$, $\hat{\ncalG}^{(1)}(s)=\exp\left(-\pi\lambda_2 \frac{(sP_2)^{2/\alpha}}{\mathrm{sinc}(2/\alpha)}\right)$, (b) is obtained by ignoring the effect of overlaps and the fact that each hole is carved out individually
from the baseline PPP $\Omega$ which leads to an upper bound for the conditional coverage probability of the typical user which is served by macro BS. In the second term of (c), 
$f(s,u)=\int_{|{u-D}|}^{{u+D}}{\arccos\left(\frac{r^2+u^2-D^2}{2ur}\right)}\frac{2\lambda_2}{1+\frac{r^\alpha}{P_2s}} r\mathrm{d}r$.
Finally, the last term in (d) follows from the PGFL of a PPP~\cite{chiu2013stochastic}
along with converting form Cartesian to polar coordinates by substituting $v=\|\nbx_1\|$ and $f(s,v)$.
The proof for the $\pc_{{1}|Z_1 < D}$ follows on the same line as the proof of $\pc_{{1}|Z_1 \geq D}$ except that $\ncalG_{1}^{(1)}(s)=\exp\!\left(\!\int_{D-z_1}^{\infty}\!\frac{-2\pi\lambda_2 r\mathrm{d}r}{1+{\frac{r^{\alpha}}{sP_2}}}\!\right)$ because the closest point of $\Phi_2$ is at least a distance $D-z_1$ away from the typical user within the hole.  
The final result follows from deconditioning over serving distance $Z_1$.
\subsection{Proof of Theorem~\ref{ThmA3}}
\label{app: D}
We consider only the closest hole in the interference field. 
The coverage probability of a typical user which connects to the nearest small cell conditioned on the fixed distances $Z_1$ and $\hat{Z_2}$, 
is lower bounded by 
\begin{align}\label{Covg_21D}
&\pc_{2|\hat{Z_2},Z_1}(s)
{\geq}\E \Bigg[\!\exp\!\Bigg(\!\!\!-s \Bigg(\sum_{\underset{\|\nbx_2\|>\hat{z_2}}{\nbx_2 \in \Omega \cap \nbb^c(\nbx_1^*,D),}}\!\!\!\!\!\!\!\!\!\!\!\! P_2 h_{\nbx_2} \|\nbx_2\|^{-\alpha}\nonumber \\&+\!\sum_{\nbx_1\in \Phi_1}\!\!\! P_1 h_{\nbx_1} \|\nbx_1\|^{-\alpha}\Bigg)\!\Bigg)\Bigg] 
\stackrel{\text{(a)}}
{=} \ncalG_{1}^{(2)}(s,\hat{z_2}) 
\nonumber\\&
\exp\Bigg(\lambda_2 \int_{\underset{\|\nbx_2\|>\hat{z_2}}{\nbb(\nbx_1^*, D),}} \frac{\mathrm{d}\nbx_2}{1+{\frac{\|\nbx_2\|^{\alpha}}{sP_2}}}\Bigg)
\E_{\Phi_1|Z_1} \Bigg[\prod_{\nbx_1 \in \Phi_\mathrm{1}}\!\!\zeta(s,\|\nbx_1\|)\Bigg]
\nonumber\\&
\stackrel{\text{(b)}}{=}
\ncalG_{1}^{(2)}(s,\hat{z_2}) \exp\left(g\left(s,z_1,\hat{z_2}\right)\right)
\zeta(s,{z_1})
\nonumber\\&
\times \exp\left(\!\!-2\pi\lambda_1 \!\int_{{z_1}}^{\infty} \!\!\left({1- \zeta(s,v)}\right)v\nrmd v\right)
\end{align}
where in step (a), $\ncalG_{1}^{(2)}(s,\hat{z_2})=\exp\left(\int_{\hat{z_2}}^{\infty}\frac{-2\pi\lambda_2 r\mathrm{d}\nbr}{1+{\frac{r^{\alpha}}{sP_2}}}\right)$, $\zeta(s,u)=\frac{1}{1+s P_1 u^{-\alpha}}$, (a) follows from the fact that the field of possible interferers involves all of the small cell BSs outside the closest hole (except the serving BS), the second term of (b) $g(s,u,\hat{z_2})=\exp\!\left({\int_{\max(\hat{z_2},|u-D|)}^{\max(\hat{z_2},u+D)}\!\arccos\left(\frac{r^2+
u^2-D^2}{2ur}\right)\frac{2\lambda_2}{1+\frac{r^{\alpha}}{sP_2}} r \nrmd r}\!\right)$ which is due to 
the fact that the field of effective interferers is out of $\nbb(0,\hat{z_2})$ that appears in the lower and upper bounds of integral expression, ${\max(\hat{z_2},|\|\nbx_1\|-D|)}$ and ${\max(\hat{z_2},\|\nbx_1\|+D)}$, and the term of  $\zeta(s,{z_1})$ denotes the effective interference from the closest macro BS at the typical point due to conditioning on $Z_1=\|\nbx_1^*\|$. 
The final result follows from deconditioning over serving distance $\hat{Z_2}$ and then over distance $Z_1$. 
\subsection{Proof of Theorem~\ref{ThmA4}}
\label{app: E}
The coverage probability of a typical user when it is served by its closest small cell (conditioned on the fixed distances $Z_1$ and $\hat{Z_2}$ is upper bounded by 
\begin{align}\label{Covg_21D}
&\pc_{2|\hat{Z_2},Z_1}(s)
{=}\E \Bigg[\!\exp\!\Bigg(\!\!-s \Bigg(\!\sum_{\nbx_2\in \Phi_2,\|\nbx_2\|>\hat{z_2}}\!\!\!\!\!\!\!\!\!\!\!\!\!\! P_2 h_{\nbx_2} \|\nbx_2\|^{-\alpha}\nonumber \\&+\!\sum_{\nbx_1\in \Phi_1}\!\!\! P_1 h_{\nbx_1} \|\nbx_1\|^{-\alpha}\Bigg)\!\Bigg)\Bigg] 
{\leq} \ncalG_{1}^{(2)}(s,\hat{z_2})
\nonumber\\&
\E_{\Phi_1|Z_1} \Bigg[\exp\Bigg(\lambda_2 \Bigg(\sum_{\nbx_1 \in \Phi_1} \int_{\underset{\|\nbx_2\|>\hat{z_2}}{\nbb(\nbx_1, D),}} \frac{\mathrm{d}\nbx_2}{1+{\frac{\|\nbx_2\|^{\alpha}}{sP_2}}}\Bigg)\Bigg)
\nonumber\\&
\prod_{\nbx_1 \in \Phi_\mathrm{1}}\!\!\frac{1}{1+s P_1 \|\nbx_1\|^{-\alpha}}\!\Bigg]
{=}
\ncalG_{1}^{(2)}(s,\hat{z_2})\times \E_{\Phi_1|Z_1} \Bigg[\nonumber\\& \prod_{\nbx_1 \in \Phi_\mathrm{1}}\exp\!\Bigg(2\lambda_2
{\int_{\max(\hat{z_2},|\|\nbx_1\|-D|)}^{\max(\hat{z_2},\|\nbx_1\|+D)}\!\frac{ \arccos\left(\frac{r^2+
\|\nbx_1\|^2-D^2}{2\|\nbx_1\|r}\right)}{1+\frac{r^{\alpha}}{sP_2}} r \nrmd r}\!\Bigg)
\nonumber\\&
\times \frac{1}{1+s P_1 \|\nbx_1\|^{-\alpha}}\Bigg]
{=}
\ncalG_{1}^{(2)}(s,\hat{z_2})\times \exp\left(g\left(s,z_1,\hat{z_2}\right)\right)
\zeta(s,{z_1})
\nonumber\\&
\times \exp\left(\!\!-2\pi\lambda_1 \!\int_{{z_1}}^{\infty} \!\!\left({1- \exp\left(g\left(s,v,\hat{z_2}\right)\right)\zeta(s,v)}\right)v\nrmd v\right)
\end{align}
where the proof follows on the same lines as Theorem~\ref{ThmA3} except that here the field of small cell interferers contains all the holes that are individually carved out from the baseline PPP, thus underestimating the interference from small cells. 
The final result follows from deconditioning over serving distance $\hat{Z_2}$ and then over distance $Z_1$ whose distribution is given by (\ref{Eq: clds00}).  
\balance
\bibliographystyle{IEEEtran}
\bibliography{SPAWC_08.bbl}  

\begin{thebibliography}{10}
\providecommand{\url}[1]{#1}
\csname url@samestyle\endcsname
\providecommand{\newblock}{\relax}
\providecommand{\bibinfo}[2]{#2}
\providecommand{\BIBentrySTDinterwordspacing}{\spaceskip=0pt\relax}
\providecommand{\BIBentryALTinterwordstretchfactor}{4}
\providecommand{\BIBentryALTinterwordspacing}{\spaceskip=\fontdimen2\font plus
\BIBentryALTinterwordstretchfactor\fontdimen3\font minus
  \fontdimen4\font\relax}
\providecommand{\BIBforeignlanguage}[2]{{%
\expandafter\ifx\csname l@#1\endcsname\relax
\typeout{** WARNING: IEEEtran.bst: No hyphenation pattern has been}%
\typeout{** loaded for the language `#1'. Using the pattern for}%
\typeout{** the default language instead.}%
\else
\language=\csname l@#1\endcsname
\fi
#2}}
\providecommand{\BIBdecl}{\relax}
\BIBdecl

\bibitem{dhillon2012modeling}
H.~S. Dhillon, R.~K. Ganti, F.~Baccelli, and J.~G. Andrews, ``Modeling and
  analysis of {K}-tier downlink heterogeneous cellular networks,'' \emph{IEEE
  Journal on Sel. Areas in Communications}, vol.~30, no.~3, pp. 550--560, 2012.

\bibitem{mukherjee2012sinr}
S.~Mukherjee, ``Distribution of downlink {SINR} in heterogeneous cellular
  networks,'' \emph{IEEE Journal on Sel. Areas in Commun.}, vol.~30, no.~3, pp.
  575--585, April 2012.

\bibitem{deng2014heterogeneous}
N.~Deng, W.~Zhou, and M.~Haenggi, ``A heterogeneous cellular network model with
  inter-tier dependence,'' in \emph{Proc., IEEE Globecom Workshops}, Dec 2014.

\bibitem{dengheterogeneousax}
------, ``Heterogeneous cellular network models with dependence,'' \emph{IEEE
  Journal on Sel. Areas in Commun.}, vol.~33, no.~10, pp. 2167--2181, Oct 2015.

\bibitem{ganti2006regularity}
R.~K. Ganti and M.~Haenggi, ``Regularity in sensor networks,'' in \emph{Proc.,
  International Zurich Seminar on Communications}.\hskip 1em plus 0.5em minus
  0.4em\relax IEEE, 2006, pp. 186--189.

\bibitem{Hybrid_ICC}
M.~Afshang, Z.~Yazdanshenasan, S.~Mukherjee, and P.~H.~J. Chong, ``Hybrid
  division duplex for {HetNets}: Coordinated interference management with
  uplink power control,'' in \emph{Proc., IEEE International Conference on
  Communications Workshops (ICC)}, 2015.

\bibitem{lee2012interference}
C.-h. Lee and M.~Haenggi, ``Interference and outage in {Poisson} cognitive
  networks,'' \emph{IEEE Trans. on Wireless Commun.}, vol.~11, no.~4, pp.
  1392--1401, April 2012.

\bibitem{yazdanshenasan2016poisson}
Z.~Yazdanshenasan, H.~S. Dhillon, M.~Afshang, and P.~H.~J. Chong, ``Poisson
  {Hole} {Process}: {Theory} and applications to wireless networks,''
  \emph{Available online: http://arxiv.org/abs/1601.01090}, 2016.

\bibitem{haenggi2005distances}
M.~Haenggi, ``On distances in uniformly random networks,'' \emph{IEEE Trans. on
  Info. Theory}, vol.~51, no.~10, pp. 3584--3586, Oct 2005.

\bibitem{YazDhiC2016}
Z.~Yazdanshenasan, H.~S. Dhillon, M.~Afshang, and P.~H.~J. Chong, ``Tight
  bounds on the {Laplace} transform of interference in a {Poisson} hole
  process,'' {\em IEEE ICC}, Kuala Lumpur, Malaysia, May 2016.

\bibitem{chiu2013stochastic}
S.~N. Chiu, D.~Stoyan, W.~S. Kendall, and J.~Mecke, \emph{Stochastic {G}eometry
  and its {A}pplications}.\hskip 1em plus 0.5em minus 0.4em\relax John Wiley \&
  Sons, 2013.

\end{thebibliography}
           
\end{document}